\documentclass[a4paper,UKenglish]{lipics-v2021}

\pdfoutput=1
\usepackage{microtype}
\usepackage[utf8]{inputenc}
\usepackage{amsmath, amsthm, amssymb, mathtools, stmaryrd}
\usepackage{relsize}
\usepackage{xspace}
\usepackage{booktabs}
\usepackage{bbm}
\usepackage{enumerate}
\usepackage{lmodern}
\usepackage{todonotes}
\usepackage{hyperref}

\mathtoolsset{centercolon}

\SetSymbolFont{stmry}{bold}{U}{stmry}{m}{n}

\makeatletter
\newcommand{\@abbrev}[3]{
  \def\c@a@def##1{
      \if ##1.
        \relax
      \else
        \@ifdefinable{\@nameuse{#1##1}}{\@namedef{#1##1}{#2##1}}
        \expandafter\c@a@def
      \fi
    }
  \c@a@def #3.
}
\@abbrev{bb}{\mathbb}{ABCDEFGHIJKLMNOPQRSTUVWXYZ}
\@abbrev{b}{\bar}{abcdefghijklmnopqrstuvwxyz}
\@abbrev{bf}{\mathbf}{ABCDEFGHIJKLMNOPQRSTUVWXYZabcdefghijklmnopqrstuvwxyz}
\@abbrev{bit}{\boldsymbol}{
ABCDEFGHIJKLMNOPQRSTUVWXYZabcdefghijklmnopqrstuvwxyz}
\@abbrev{mc}{\mathcal}{ABCDEFGHIJKLMNOPQRSTUVWXYZ}
\@abbrev{mb}{\mathbb}{ABCDEFGHIJKLMNOPQRSTUVWXYZ}
\@abbrev{mf}{\mathfrak}{
ABCDEFGHIJKLMNOPQRSTUVWXYZabcdefghijklmnopqrstuvwxyz}
\@abbrev{rm}{\mathrm}{ABCDEFGHIJKLMNOPQRSTUVWXYZabcdefghijklmnopqrstuvwxyz}
\@abbrev{scr}{\mathscr}{
ABCDEFGHIJKLMNOPQRSTUVWXYZabcdefghijklmnopqrstuvwxyz}
\@abbrev{sf}{\mathsf}{ABCDEFGHIJKLMNOPQRSTUVWXYZabcdefghijklmnopqrstuvwxyz}
\makeatother

\newcommand{\N}{\mbN}
\newcommand{\Z}{\mbZ}

\newcommand{\field}{\mathbb}

\renewcommand{\phi}{\varphi}
\renewcommand{\theta}{\vartheta}
\newcommand{\tup}{\bar}

\newcommand{\CFIgraph}[3]{\text{\sf CFI}\,[#1, #2, #3]}
\newcommand{\CFIgraphO}[3]{\text{\sf CFI}_{\text{\sf O}}\,[#1, #2, #3]}
\newcommand{\CFIgraphI}[3]{\text{\sf CFI}_{\text{\sf I}}\,[#1, #2, #3]}

\newcommand{\Primes}{\mbP}

\newcommand{\FPC}{\textup{FPC}\xspace}

\newcommand{\cequivx}[1]{\ensuremath{\equiv^{#1}}}
\newcommand{\cequivk}{\cequivx k}

\newcommand{\simequivx}[2]{\ensuremath{\equiv^\text{IM}_{#1, #2}}}

\newcommand{\IMequiv}{\simequivx}

\newcommand{\FPR}{\mathrm{FPR}}
\newcommand{\LALogic}{\ensuremath{\mathrm{LA}^{\omega}}}
\newcommand{\LAkLogic}{\ensuremath{\mathrm{LA}^{k}}}
\newcommand{\iso}{\ensuremath{\cong}}

\newcommand{\primes}{\mathbb{P}}

\newcommand{\nats}{\mathbb{N}}
\newcommand{\ra}{\rightarrow}

\newcommand{\set}[1]{\lbrace #1 \rbrace}
\newcommand{\setcond}[2]{\lbrace #1 : #2 \rbrace}
\newcommand{\StructO}{A}
\newcommand{\StructI}{B}
\newcommand{\neighbors}[2]{N_{#1}(#2)}
\newcommand{\rel}{R}

\title{Limitations of the Invertible-Map Equivalences}
 
\author{Anuj Dawar}{University of Cambridge, UK}{anuj.dawar@cl.cam.ac.uk}{https://orcid.org/0000-0003-4014-8248}{}
\author{Erich Grädel}{RWTH Aachen University, Germany}{graedel@logic.rwth-aachen.de}{https://orcid.org/0000-0002-8950-9991}{}
\author{Moritz Lichter}{TU Darmstadt, Germany}{lichter@mathematik.tu-darmstadt.de}{https://orcid.org/0000-0001-5437-8074}{The research leading to these results has received funding from the European Research Council (ERC) under the European Union’s Horizon 2020 research and innovation programme (EngageS: grant agreement No.\ 820148).}

\authorrunning{A. Dawar, E. Gr\"adel, and M.~Lichter} 
 
\Copyright{Anuj Dawar, Erich Grädel, and Moritz Lichter} 

\nolinenumbers 
\hideLIPIcs  

\begin{document} 
\maketitle

\begin{abstract}
This note draws conclusions that arise by combining two recent papers,
by Anuj Dawar, Erich Grädel, and Wied Pakusa, published at ICALP 2019 and by Moritz Lichter, published at LICS 2021. 
In both papers, the main technical results rely on the combinatorial and algebraic analysis of the invertible-map 
equivalences $\IMequiv{k}{Q}$ on certain variants of Cai-Fürer-Immerman structures (CFI-structures for short).
These $\IMequiv{k}{Q}$-equivalences, for a natural number $k$ and a set of  primes~$Q$, 
refine the well-known Weisfeiler-Leman  equivalences used in algorithms for graph isomorphism.
The intuition is that two graphs $G\IMequiv{k}{Q}H$ cannot be distinguished by  iterative refinements of 
equivalences on $k$-tuples defined via linear operators on vector spaces over fields of characteristic $p \in Q$.

In the first paper  it has been shown, using considerable algebraic machinery, that for a prime $q \notin Q$, 
the  $\IMequiv{k}{Q}$ equivalences are not strong enough to distinguish between non-isomorphic
CFI-structures over the field $\field F_q$. In the second paper, a similar but not identical construction for 
CFI-structures over the rings $\mbZ_{2^i}$ has, again by rather involved combinatorial and algebraic arguments,
been shown to be indistinguishable with respect to $\IMequiv{k}{\{2\}}$. Together with earlier work on rank logic,
this second result suffices to separate rank logic from polynomial time.

We show here that the two approaches can be unified to prove that  CFI-structures
over the rings $\mbZ_{2^i}$ are in fact  indistinguishable with respect to
$\IMequiv{k}{\primes}$, for the set $\primes$ of \emph{all} primes. In particular, this implies the following two results.
\begin{itemize}
\item There is no fixed $k$ such that the invertible-map equivalence $\IMequiv{k}{\primes}$ coincides with isomorphism on all finite graphs.
\item No extension of fixed-point logic by linear-algebraic operators over fields can capture polynomial time.
\end{itemize}
\end{abstract}


\section{Invertible-map equivalences and linear algebraic logics}

Invertible-map equivalences are refinements of the Weisfeiler-Leman method, an important technique in the study of the 
graph isomorphism problem.  For each positive integer $k$, the $k$-dimensional Weisfeiler-Leman method ($k$-WL method for short) defines an equivalence relation $\cequivk$ which over-approximates isomorphism in the sense that if $G\iso H$ for a pair of graphs $G$ and $H$, then $G \cequivk H$ for any $k$. 
These equivalence relations get finer with increasing~$k$ and approach isomorphism in the limit. Indeed, if $G$ and $H$ are $n$-vertex graphs then $G \cequivx{n} H$ if, and only if, $G \iso H$ and, for each fixed $k$, the equivalence relation $\cequivk$ is decidable in time $n^{O(k)}$.  
Thus, if there was a fixed $k$ such that $\cequivk$ was the same as isomorphism, we would have a polynomial-time algorithm for graph isomorphism.  
However, there is no such fixed $k$.  Cai, F\"urer, and Immerman~\cite{CFI92} showed that there are pairs of non-isomorphic  graphs $G$ and $H$ with $O(k)$ vertices such that $G \cequivk H$.  We call the construction of such graphs the CFI-construction.
The Weisfeiler-Leman equivalences are also of central importance in descriptive complexity theory since they delimit the power of logics with counting operators, such as
fixed-point logic with counting ($\FPC$), which is  a fundamental formalism in the quest for a logic for PTIME (see~\cite{Grohe08}).

The CFI-construction, in its original form, can be seen as a graph encoding of linear equation systems  over 
the field $\field{F}_2$ \cite{AtseriasBulDaw09}. Thus, while $\FPC$ is not strong enough to tell apart 
non-isomorphic CFI-structures, this can be done by stronger extensions of fixed-point logics that are powerful enough to solve such equation systems. A number of such extensions have been studied in \cite{DawarGraHolKopPak13};
the most influential one is \emph{rank logic} (FPR), proposed in \cite{DawarGroHolLau09}.
Rank logic extends fixed-point logic by operators for the rank of definable matrices over a given finite field $\field F_p$. For a somewhat more
powerful variant of rank logic $\FPR^*$, studied in~\cite{GraedelPak19}, it has until recently been open
whether it defines all polynomial-time properties  of finite structures.

The invertible-map equivalences have been defined in~\cite{DawarHol17} as a tool to study the 
expressive power of rank logic.  Like the $k$-WL equivalences, they are defined by iterated refinements
of equivalences between $k$-tuples. However,  the refinement process is not defined on the basis of counting, but 
on the basis of invertible maps between matrices obtained from the given tuples by appropriate substitutions.
For a formal definition, we refer to  \cite[Sect. 3.1]{DawarGraPak19a}.
The equivalences  $\IMequiv{k}{\{2\}}$ properly refine the Weisfeiler-Leman equivalences in 
the sense that $G \IMequiv{k'}{\{2\}} H$ for sufficiently large $k'$ implies $G \cequivk H$ for all graphs $G$ and $H$,
but for the pairs $G, H$ obtained in the CFI-construction, $G \not\IMequiv{3}{\{2\}} H$. 
As shown in \cite{DawarHol17} there is, for every formula~$\phi$ of rank logic FPR, 
a $k\in\N$ and a finite set $Q$ of primes such that the class of models of $\phi$ is closed under $\IMequiv{k}{Q}$.
But in fact, the invertible-map equivalences are potentially much finer
than the equivalences under rank logic.
They delimit the expressive power not just of rank logic, but of arbitrary extensions of fixed-point logic by
linear-algebraic operators. Intuitively, a linear-algebraic operator over a field $\field F$
is any function $f$  that maps tuples $(M_1, \dots, M_m)$ of 
$\field F$-linear transformations on (subspaces of) 
an abstract vector space $\field F^\mcB$ to some kind of linear-algebraic information 
$f(M_1, \dots, M_m) \in \N$.
We do not even require that the function $f$ is computable, but to define
``linear-algebraic information'' it has to be invariant under $\field F$-vector space isomorphisms. 
This means that $f(M_1, \dots, M_m)=f(N_1, \dots, N_m)$
for any two sequences $(M_1, \dots, M_m)$ and $(N_1, \dots, N_m)$
that are simultaneously similar, in the sense that there is a $\field F$-vector space isomorphism
$S$ such that $N_i \cdot S = S \cdot M_i$ for all $i\leq m$.
The general linear-algebraic logics  $\LAkLogic(Q)$, defined in \cite{DawarGraPak19}, are
infinitary $k$-variable logics with generalized quantifiers
for all linear-algebraic operators over finite vector spaces of characteristic $p \in Q$.
For a detailed definition that is not needed here we refer to \cite[Sect.~3.2.]{DawarGraPak19a}.

Notice that the logics  $\LAkLogic(Q)$ and $\LALogic(Q) = \bigcup_{k \in \omega} \LAkLogic(Q)$
are non-effective, infinitary logics that are not intended for practical use.
Their relevance stems from the fact that they encompass any extension of 
first-order logic or  fixed-point logics by means of $Q$-linear-algebraic operators.
Thus, inexpressibility results for  $\LAkLogic(Q)$ and $\LALogic(Q)$ directly translate to
inexpressibilty results for all such logics, in particular for  rank logic or logics with solvability operators for
linear equation systems.

It has been shown in \cite{DawarGraPak19} that $\LAkLogic(Q)$ is the logic for
which the invertible-map equivalence~$\IMequiv{k}{Q}$ is the natural notion of elementary equivalence.  
\begin{theorem}
	\label{TheoremLKequivIM}
   Let $k \geq 2$ be a positive integer and $Q$ a set of prime
   numbers.  For any finite structure $\mfA$ and $\tup{a},\tup{b} \in
   A^k$, the following are equivalent:
   \begin{enumerate}
  \item  $(\mfA, \tup a)  \IMequiv{k}{Q}  (\mfA, \tup b)$; and
  \item  for every formula
  $\phi$ of $\LAkLogic(Q)$, $\mfA \models \phi[\tup a]$ if, and only if, $\mfA \models \phi[\tup b]$.
   \end{enumerate}
\end{theorem}

\section{Invertible-map equivalences for generalised CFI-structures}

We next present a high-level exposition of the results in \cite{DawarGraPak19} and \cite{Lichter21}
on invertible-map equivalences of CFI-structures, and their consequences for graph isomorphism and 
descriptive complexity. We refer to the full versions of these papers, published on ArXiv \cite{DawarGraPak19a, Lichter21a}.

It is well-known that the  CFI-construction can be adapted beyond the field $\field F_2$
to many other algebraic structures. A general variant
due to Holm \cite{Holm10} is based on arbitrary finite Abelian groups.
In \cite{GraedelPak19} a variant over prime fields 
$\field F_p$ has been used to show that formulae of $\FPR$
that do not use a rank operator over the field $\field F_p$ 
are no more expressive than formulae of $\FPC$ over these graphs.
This separates the expressive power of $\FPR$ from that of $\FPR^*$, and proves that $\FPR$ does not capture
PTIME.  
In \cite{DawarGraPak19a}, the same graph construction
has been analysed with significantly deeper algebraic machinery,
connecting it to invertible-map equivalences for primes $p \notin Q$.

More precisely, this variant of the CFI-construction associates with 
every connected, $3$-regular, ordered, and simple \emph{base graph} 
$G=(V,E,\leq)$, every prime field $\field F_p$, and every function
$\lambda: V\ra\field F_p$
a CFI-structure $\CFIgraph{G}{\field F_p}{\lambda}$, with the following properties:
\begin{itemize}
\item The automorphism group of  $\CFIgraph{G}{\field F_p}{\lambda}$ is an elementary Abelian $p$-group.
\item Two CFI-structures $\CFIgraph{G}{\field F_p}{\lambda}$ and $\CFIgraph{G}{\field F_p}{\sigma}$ 
over the same base graph $G$ are isomorphic if, and only if, 
$\sum \lambda = \sum_{v \in V} \lambda (v) = \sum_{v \in V} \sigma(v) = \sum \sigma$.
\end{itemize}

The \emph{CFI-problem} (over a class $\mcF$ of base graphs 
and a field $\field F_p$) is to 
decide, given a structure 
$\CFIgraph{G}{\field F_p}{\lambda}$ with $G\in\mcF$, whether $\sum\lambda = 
0$. The CFI-problem is solvable in polynomial time, for instance by Gaussian elimination.

\medskip
For proving logical inexpressibility results, the full power of the CFI-construction is unfolded 
when the graphs in the underlying class $\mcF$
are highly connected. The class used in \cite{DawarGraPak19a} is a family 
$\mcF = \{ G_n : n \in \mathbb N\}$ of 3-regular, connected expander graphs 
where $G_n$ has $\mcO(n)$ vertices. By the Cai-Fürer-Immerman Theorem \cite{CFI92}
and its well-known generalisations to other algebraic structures than $\field F_2$, we
have the following property:
\begin{itemize}
\item For every $G_n\in\mcF$ and all $\lambda,\sigma:V\ra\field F_p$ we have that 
$\CFIgraph{G_n}{\field F_p}{\lambda} \equiv^{\Omega(n)} 
\CFIgraph{G_n}{\field F_p}{\sigma}$.
\end{itemize}

A final important fact about these CFI-structures is a \emph{homogeneity} property: 
Despite the fact that counting logic cannot determine the 
full isomorphism type of a CFI-structure, it can, with $\mcO(k)$ many variables,
distinguish between those pairs of $k$-tuples which are not related via an automorphism of the CFI-structure.

\begin{itemize}
\item For all $k$-tuples $\tup a,\tup b$ in a CFI-structure $\mfA = \CFIgraph{G}{\field F_p}{\lambda}$ with $G\in\mcF$, 
we have that $(\mfA, \tup a) \equiv^{3k} (\mfA, \tup b)$ if, and only if,  $f(\tup a)=\tup b$ for some automorphism $f$
of $\mfA$.
\end{itemize}
 
Based on these properties, and on methods from the representation theory of finite groups,
such as Maschke's Theorem, the main technical result of \cite{DawarGraPak19a} says the following:
on CFI-structures for~$\mcF$ and the field $\field F_p$
the distinguishing power of $\IMequiv{k}{Q}$, where $p \not\in Q$, is no greater than
the counting equivalence $\cequivx{\ell}$ for some fixed $\ell$.

\begin{theorem} \label{DGPtheorem}
  Let $p\not\in Q$.  For every $k$
	there is an $n$ such that for every $G_m \in \mcF$ 
	satisfying $m\geq n$
	and all $\lambda,\sigma$ we have that 
$\CFIgraph{G_m}{\field F_p}{\lambda}  \IMequiv{k}{Q} \CFIgraph{G_m}{\field F_p}{\sigma}$.
\end{theorem}

\begin{corollary}
  If $Q\neq \Primes$, there is no fixed $k$ such that $\IMequiv{k}{Q}$
  coincides with isomorphism on all finite structures.
\end{corollary}

The interesting question left open by this result is, of course, 
the case when $Q = \Primes$.
Since the CFI-problem, for arbitrary base graphs, is solvable
in polynomial time by solving systems of
linear equations, we get the following limitations for the expressive power of the
logics $\LALogic(Q)$. 

\begin{corollary}\label{cor:undefinable}
   If $Q\neq \Primes$, there is a class of finite structures  that is decidable in polynomial time,
   but not definable in $\LALogic(Q)$. 
\end{corollary}

Since $\LALogic(Q)$ subsumes $\FPC$, no extension of fixed-point
logic by $Q$-linear algebraic operators can capture PTIME,
unless it includes such operators for all prime characteristics.

\medskip
More recently,  a somewhat different CFI-construction over the rings $\Z_{2^i}$ has been used 
by Lichter \cite{Lichter21a} to separate rank logic from PTIME. 
His construction of CFI-structures $\CFIgraph{G}{\Z_{2^i}}{\lambda}$ is not based on 3-regular
graphs, but on highly connected regular graphs of large degree and girth. Further, but this is a
minor point, the last component is not a function on vertices, but
a function $\lambda: E\ra \Z_{2^i}$ defining the values by which edges are twisted.
Analogous properties as above apply. In particular, 
\begin{itemize}
\item The automorphism group of $\CFIgraph{G}{\Z_{2^i}}{\lambda}$ is an Abelian $2$-group.
\item Two CFI-structures $\CFIgraph{G}{\Z_{2^i}}{\lambda}$ and $\CFIgraph{G}{\Z_{2^i}}{\sigma}$ 
are isomorphic if, and only if, 
$\sum \lambda = \sum_{e \in E} \lambda (e) = \sum_{e \in E} \sigma(e) = \sum \sigma$.
\end{itemize}

The analysis of these CFI-structures is done in terms of the game-theoretic description of
the invertible-map equivalences, the so-called invertible-map game introduced in \cite{DawarHol17},
using combinatorial objects called blurers. 
The main technical result of \cite{Lichter21a} shows that these CFI-structures
cannot be told apart by invertible-map equivalences for the prime 2.

\begin{theorem} \label{LichterTheorem}
For each $k$ there exists a graph $G=(V,E,\leq)$, a number $i$, and two functions
$\lambda,\sigma: E\ra \Z_{2^i}$ such that $\sum\sigma=\sum \lambda + 2^{i-1}$ and 
$\CFIgraph{G}{\Z_{2^i}}{\lambda}\IMequiv{k}{\{2\}}\CFIgraph{G}{\Z_{2^i}}{\sigma}$.
\end{theorem} 

Further, Lichter refines an argument from \cite{GraedelPak19} to
show that on the CFI-structures over~$\Z_{2^i}$, 
every formula of $\FPR^*$ is equivalent to an $\FPR$ formula
with rank operators only over the field $\field F_2$.
But these cannot tell apart $\IMequiv{k}{\{2\}}$-equivalent structures.
Thus, there exists a variant of the CFI-problem that is not definable in rank logic.

\begin{corollary}
$\FPR^*$ does not capture PTIME.
\end{corollary}

\section{Combining the constructions}
\label{sec:combine-constructions}

To combine the results of \cite{DawarGraPak19a} and \cite{Lichter21a} we want to show that
the CFI-structures $\CFIgraph{G}{\Z_{2^i}}{\lambda}$ are not just  $\IMequiv{k}{\{2\}}$-equivalent
but in fact $\IMequiv{k}{\primes}$-equivalent for the set of \emph{all} primes $\primes$.
For this, we have to show that the differences in the two CFI-constructions do not really 
matter.

Both CFI-structures are based on the well-known CFI-gadgets.
These gadgets originally consist of inner and outer vertices.
Every outer vertex is adjacent to some inner vertices in the gadget.
Two gadgets are connected by connecting their corresponding outer vertices.
For $d$-regular graphs, the inner vertices can be replaced by
$d$-ary relations, which is done in~\cite{DawarGraPak19a}.
Alternatively, \cite{Lichter21a} leaves out the outer vertices
and directly connects the inner vertices,
which is important to yield structures of the same signature
for different degrees of the base graph.
When using only one sort of vertices (so either only inner or only outer ones)
fewer case distinctions are needed.

For a simple and connected base graph $G=(V,E,\leq)$
and a function $\lambda \colon E \to \Z_{2^i}$
we define the two constructions
$\CFIgraphO{G}{\Z_{2^i}}{\lambda}$ using only outer vertices and
$\CFIgraphI{G}{\Z_{2^i}}{\lambda}$ using only inner vertices, respectively.

\subparagraph{Construction using outer vertices}
This construction requires that $G$ is $d$-regular.
For each vertex $u \in V$ with neighbourhood 
$\neighbors{G}{u} = \set{v_1, \dots, v_d}$ we define a gadget
consisting of vertices $\StructO_u := \Z_{2^i} \times \neighbors{G}{u}$ and two relations:
\begin{align*}
	\rel_{u} &:= \setcond{
		((a_1,v_1), \dots, (a_d,v_d)) \in \StructO_u^d}
		{\sum_{i=1}^d a_i = 0}, &u \in V,\\
	C_{u} &:= \setcond{((a,v),(a+1,v)) \in \StructO_u^2}{a\in \Z_{2^i}, v \in \neighbors{G}{u}}, &u \in V.
\end{align*}
The \emph{CFI-relation} $\rel_u$
connects all $d$-tuples of vertices for each neighbour with sum $0$ (in~$\Z_{2^i}$)
and the \emph{cycle relation}
realizes the automorphism group $\Z_{2^i}$
on the vertices of each neighbour of~$u$.
We obtain the CFI-structure $\CFIgraphO{G}{\Z_{2^q}}{\lambda} := (\StructO, \rel, C, I, \preceq)$ as follows:
The universe~$\StructO$ is given by the disjoint union of the $\StructO_u$ for all $u \in V$,
and likewise $\rel := \bigcup_{u \in V} \rel_u$
and $C:= \bigcup_{u \in V} C_{u}$.
The \emph{inverse relation} pairs
additive inverses for each edge (shifted by $\lambda$):
\begin{align*}
	I &:= \setcond{((a, v),(b,u)) \in \StructO_u \times \StructO_v}{\set{u,v} \in E, a + b = \lambda(\set{u,v})}.
\end{align*}
Finally, the preorder $\preceq$
is just the extension of $\leq$ to the gadgets:
for $(a,u') \in \StructO_u$ and $(b,v') \in \StructO_v$
we have $(a,u') \preceq (b,v')$ if
$(u,u')$ is lexicographically smaller than $(v,v')$.

\subparagraph{Construction using inner vertices}
To define $\CFIgraphI{G}{\Z_{2^i}}{\lambda}$
we replace the $d$-ary relation~$R$ with vertices
and thus can omit the restriction to a fixed degree.
For each vertex $u \in V$  we define a gadget
consisting of vertices $\StructI_u$ and two families of relations:
\begin{align*}
	\StructI_u &:= \setcond{\tup{a} \in \Z_{2^i}^{\neighbors{G}{u}}}{\sum \tup{a} = 0}, & u \in V,\\
	N_{u,v} &:= \setcond{
		(\tup{a},\tup{b}) \in \StructI_u^2} {\tup{a}(v) = \tup{b}(v)}, &u \in V, v\in \neighbors{G}{u},\\
	C_{u,v} &:= \setcond{(\tup{a},\tup{b}) \in \StructI_u^2}{\tup{a}(v)+1=\tup{b}(v)}, &u \in V, v\in \neighbors{G}{u}.
\end{align*}
The relation $N_{u,v}$ identifies a set of vertices in $\StructI_u$
corresponding to the vertex $(a,v) \in \StructO_u$.
For every $v \in \neighbors{G}{u}$ and $a \in \Z_{2^i}$
a clique is added between the vertices in the set corresponding to $(a,v)$.
These cliques are a partition of $\StructI_u$ for a fixed $v$.
The other relation $C_{u,v}$
represents the relation $C$
by adding directed complete bipartite graphs between subsequent cliques.
We need different relations for every neighbour $v \in \neighbors{G}{u}$ because the relations overlap.

We obtain the CFI-structure $\CFIgraphI{G}{\Z_{2^q}}{\lambda} := (\StructI, N, C, I, \preceq)$ as follows:
the universe~$\StructI$ is given by the disjoint union of all $\StructI_u$ for all $u \in V$.
The relations $N$ and $C$ are $4$-ary equivalence relations on pairs,
such that the $N_{u,v}$ respectively $C_{u,v}$ are given as union of
equivalence classes:
\begin{align*}
	N &:= \setcond{(\tup{a},\tup{b}, \tup{a}',\tup{b}')}{
	\setcond{(u,v)}{(\tup{a},\tup{b}) \in N_{u,v}}
	\leq 
	\setcond{(u',v')}{(\tup{a}',\tup{b}') \in N_{u',v'}}
}.
\end{align*}
Here we extended $\leq$ to sets of pairs of vertices in the base graph.
The relation $C$ is defined similarly.
The preorder $\preceq$ is again the preorder obtained as the
lexicographical extension of $\leq$ to the vertices in $B_u$.
Now connecting gadgets becomes similar to the case of $\CFIgraphO{G}{\Z_{2^q}}{\lambda}$.
Instead of adding an edge between two vertices in 
$\StructO_u \times \StructO_v$,
we add complete bipartite graphs
between the corresponding vertices in $\StructI_u \times \StructI_v$:
\begin{align*}
	 I &:= \setcond{(\tup{a},\tup{b}) \in \StructI_u\times \StructI_v}{
	 	\set{u,v} \in E, \tup{a}(v) + \tup{b}(u) =  \lambda(\set{u,v}) }.
\end{align*}

For easier presentation, the two structures 
$\CFIgraphI{G}{\Z_{2^q}}{\lambda}$
and $\CFIgraphO{G}{\Z_{2^q}}{\lambda}$
still differ slightly from the ones in~\cite{DawarGraPak19a}
and~\cite{Lichter21a}.
In~\cite{DawarGraPak19a} functions $\lambda \colon V \to \field{F}_p$
instead of $\lambda \colon E \to \Z_{2^i}$ are used.
This results in isomorphic structures.
In~\cite{Lichter21a} more relations apart from $I$ are added
to make local isomorphism types more informative.
Nevertheless, these structures have the same automorphisms
and in fact the additional relations are definable in $3$-variable logic using the relation~$I$.

The $k$-orbits of a CFI-structure $\CFIgraphI{G}{\Z_{2^i}}{\lambda}$ over a $(k+3)$-connected base graph $G=(V,E,\leq)$ 
can be defined in $(k+2)$-variable counting logic.
The proof is analogous to the one
in~\cite{GraedelPak19} for the
case of $\field{F}_p$ instead of $\Z_{2^i}$.

\subparagraph{Combining results.}
Our ultimate goal is to prove the following theorem:
\begin{theorem}\label{TheoremLimitationIMEquiv}
	For each $k$ there exists a graph $G=(V,E,\leq)$, a number $i$, and two functions
	$\lambda,\sigma\colon E\ra \Z_{2^i}$ such that $\sum\sigma=\sum \lambda + 2^{i-1}$ and 
	$\CFIgraphI{G}{\Z_{2^i}}{\lambda}\IMequiv{k}{\Primes}\CFIgraphI{G}{\Z_{2^i}}{\sigma}$.
\end{theorem} 

This theorem is proved by combining the proofs of Theorem~\ref{DGPtheorem} and
Theorem~\ref{LichterTheorem}.  Specifically, a close look at the base graphs used to prove
Theorem~\ref{LichterTheorem} in~\cite{Lichter21a} immediately gives the
following.
\begin{lemma}\label{LemmaLichterDetails}
	For each $k$ there exist $c$, $d$, $g$,~and $i$
	such that for every regular base graph $G=(V,E,\leq)$
	of degree at least $d$, vertex-connectivity at least $c$,
	and girth at least $g$ there are 
	functions $\lambda,\sigma\colon E\ra \Z_{2^i}$ such that $\sum\sigma=\sum \lambda + 2^{i-1}$ and 
	$\CFIgraphI{G}{\Z_{2^i}}{\lambda}\IMequiv{k}{\{2\}}\CFIgraphI{G}{\Z_{2^i}}{\sigma}$.
\end{lemma}

Our aim is to argue that the case of primes other than $2$ can be
covered by the methods used to prove  Theorem~\ref{DGPtheorem}
in~\cite{DawarGraPak19a}.  Specifically, we examine the properties of
the CFI-structures used in that proof and argue that they are
(sufficiently) satisfied by the alternative CFI-structures defined
here.
The proofs in Section~8 in~\cite{DawarGraPak19a} depend on the following
properties of CFI-structures:
\begin{itemize}
\item \emph{Homogeneity}:
	A structure is called $\ell$-homogeneous,
	if for every $t$ the $t$-orbits of the structure can be
        defined in counting logic with $\ell\cdot t$ variables.  The
        proof in~\cite{DawarGraPak19a} relies on the fact (proved
        in~\cite{GraedelGroPagPak19}) that if the base graph is a
        $3$-regular expander, then the resulting CFI-structures are
        $\ell$-homogeneous for some fixed value of $\ell$.  However,
        the construction we are using here uses base graphs that are
        $d$-regular (for increasing values of~$d$) and not necessarily
        expanders.  The proof of Lemma~\ref{LemmaLichterDetails}
        relies on a weaker connectivity assumption: that the graphs
        are $c$-connected.  With this, we cannot prove that the
        structure $\CFIgraphI{G}{\Z_{2^i}}{\lambda}$ is homogeneous.
        However, we can show that the $t$-orbits for $t \leq c$ are
        definable in counting logic with no more than $\ell \cdot t$
        variables for some constant $\ell$.

        Homogeneity is used in the proof of Theorem~8.2
        in~\cite{DawarGraPak19a} to construct a formula of counting
        logic ordering the $t$-orbits of CFI-structures.  It is clear
        from the proof of the theorem that we need this only for
        values of $t$ not exceeding $k$, the number of variables for
        which we aim to establish equivalence in
        Theorem~\ref{DGPtheorem}.  Thus, the full strength of
        homogeneity is not necessary.  We can choose, for any $k$,
        base graphs with sufficiently large values of $d$, $c$ and $g$
        as in Lemma~\ref{LemmaLichterDetails} and in these, $t$-orbits
        for all values of $t \leq c$ can be ordered in counting logic.
      \item \emph{Structure of automorphism groups:}
        It is pointed out in~\cite{DawarGraPak19a} that the
        automorphism groups of the CFI-structures
        $\CFIgraph{G}{\field F_p}{\lambda}$ constructed there are
        elementary Abelian $p$-groups.  For our structures
        $\CFIgraphI{G}{\Z_{2^i}}{\lambda}$, the automorphism groups
        are Abelian $2$-groups but not necessarily elementary.
        However, the proof of Theorem~\ref{DGPtheorem} does not use
        the assumption of elementariness anywhere.  The fact that it
        is an Abelian $p$-group is sufficient.
	\item \emph{Automorphisms as ordered objects:}
	Part~(3) of the proof of Theorem~8.2 in~\cite{DawarGraPak19a}
	exploits that automorphisms of CFI-structures can be
	represented as ordered objects.
	This works for $\CFIgraphI{G}{\Z_{2^i}}{\lambda}$
	exactly as for $\CFIgraphO{G}{\Z_{2^i}}{\lambda}$
	by exploiting the total order on the vertices
	(and hence of the edges) of $G$:
	for every edge it is stored by which amount an
	automorphism twists the edge.
\end{itemize}

Thus, we have seen that both CFI-constructions satisfy
the same crucial properties, which permits us to establish the following lemma.
\begin{lemma}\label{LemmaCFICoprimeField}
	For every $k$ there exists a number $c$ such that for every
	$c$-connected base graph $G=(V,E,\leq)$
	and every $\lambda,\sigma \colon E \to \Z_{2^i}$
	it holds that
	$\CFIgraphI{G}{\Z_{2^i}}{\lambda}\IMequiv{k}{\Primes \setminus \set{2}}\CFIgraphI{G}{\Z_{2^i}}{\sigma}$.
\end{lemma}

So we  know that for every $k$ there is a pair $\mfA, \mfB$ of CFI-structures satisfying 
$\mfA \IMequiv{k}{\set{2}} \mfB$ and $\mfA \IMequiv{k}{\primes \setminus \set{2}} \mfB$.
To combine these results, we show in general that
if  $\mfA \IMequiv{k+3}{Q} \mfB$
and $\mfA \IMequiv{k+3}{P} \mfB$,
then $\mfA \IMequiv{k}{P\cup Q} \mfB$.
This is not immediate, because it is not clear
whether nesting linear-algebraic operators of characteristics in $P$ and $Q$
increases the distinguishing power.

\begin{lemma}
	\label{LemmaCombineCharacteristics}
	Let $P, Q$ be two set of primes,
	$k,i \in \nats$, $G=(V,E,\leq)$ be a $(k+3)$-connected base graph,
	$\lambda, \sigma \colon E \to \Z_{2^i}$,
	$\mfA = \CFIgraphI{G}{\Z_{2^i}}{\lambda}$, and
	$\mfB = \CFIgraphI{G}{\Z_{2^i}}{\sigma}$.
	If $\mfA \IMequiv{k+3}{P} \mfB$ and $\mfA \IMequiv{k+3}{Q} \mfB$,
	then $\mfA \IMequiv{k}{P\cup Q} \mfB$.
\end{lemma}
\begin{proof}
	We say that two tuples $\tup{a} \in A^n$ and $\tup{b} \in B^n$
	(for some $n \leq k$)
	have the same type,
	if the same $(k+2)$-variable counting logic formula
	defines the orbit of $\tup{a}$ and $\tup{b}$.
	We show by induction on formulae
	that for every $\LAkLogic(P \cup Q)$ formula $\phi$
	and every $\tup{a} \in A^k$ and $\tup{b} \in B^k$
	that have the same type it holds that $\mfA \models \phi[\tup{a}]$
	if, and only if, $\mfB \models \phi [\tup{b}]$.
	
	The only interesting case is the one of a linear-algebraic operator $f$
	of characteristic $p \in P \cup Q$.
	Assume without loss of generality that $p \in P$.
	For simplicity, we denote the generalized quantifier%
	\footnote{Formally, $\LAkLogic(P\cup Q)$ uses an interpretation instead of $\ell$ many plain formulae, but the argument remains the same. For details, we refer to~\cite{DawarGraPak19a}.}
	corresponding to $f$ by $\psi = \mathcal{Q}_f^{m,t}(\phi_1, \dots, \phi_\ell) $ for $t \in \nats$ and $2m \leq k$.
	Here $\phi_1, \dots, \phi_\ell$ are $\LAkLogic(P \cup Q)$ formulae,
	where $2m$ variables are bound by the quantifier.
	These formulae correspond to $0/1$ $A^m \times A^m$ matrices
	$M_1, \dots, M_\ell$
	and likewise to 
	$B^m \times B^m$ matrices
	$N_1, \dots, N_\ell$ (for details we refer to~\cite{DawarHol17}).
	The generalized quantifier has $n \leq k - 2m$ free variables
	and is satisfied if $f(M_1, \dots , M_\ell) \geq t$.
	
	Let $\tup{a} \in A^n$ and $\tup{b} \in B^n$ have the same type.
	By induction hypothesis $\mfA \models \phi_i[\tup{a}\tup{a}']$
	if, and only if, $\mfB \models \phi_i [\tup{b}\tup{b}']$
	for every $i \in \ell$, $\tup{a}'\in A^{2m}$, and $\tup{b}' \in B^{2m}$
	such that $\tup{a}\tup{a}'$ and $\tup{b}\tup{b}'$
	have the same type.
	That is,
	there are $(k+2)$-variable counting logic formulae $\phi'_1, \dots , \phi'_\ell$
	equivalent to the $\phi_i$ on $\mfA$ and $\mfB$
	(namely the disjunction of all orbit-defining formulae
	for all orbits satisfying $\phi_i$).
	With an additional free variable we can simulate counting
	and obtain equivalent  $\phi''_1, \dots , \phi''_\ell$ $\mathrm{LA}^{k+3}(P)$ formulae.
	
	Then $\psi' := \mathcal{Q}_f^{m,t}(\phi''_1, \dots, \phi''_\ell)$
	is an $\mathrm{LA}^{k+3}(P)$ formula
	equivalent to $\psi$ on~$\mfA$ and~$\mfB$.
	For the sake of contradiction,
	assume without loss of generality that $\mfA \models \psi'[\tup{a}]$
	but $\mfB \not\models \psi'[\tup{b}]$.
	Let $\chi[\tup{x}]$ be the $(k+2)$-variable logic formula defining the orbits of $\tup{a}$ and $\tup{b}$
	and $\chi'[\tup{x}]$ be the equivalent $\mathrm{LA}^{k+3}(P)$ formula.
	Then the $\mathrm{LA}^{k+3}(P)$ sentence
	$\forall \tup{x}.~\chi'[\tup{x}] \Rightarrow \psi'[\tup{x}]$
	distinguishes~$\mfA$ and $\mfB$.
	But by assumption and Theorem~\ref{TheoremLKequivIM} such a sentence does not exist.
\end{proof}
We believe that with a more careful analysis
the decrease of the number of variables from $k+3$
for $P$ and $Q$
to $k$ for $P \cup Q$ in Lemma~\ref{LemmaCombineCharacteristics} is not needed.
Finally, we are ready to prove Theorem~\ref{TheoremLimitationIMEquiv}.

\begin{proof}[Proof of Theorem~\ref{TheoremLimitationIMEquiv}]
	For every $d\geq 2$ and $g\geq 3$
	there exists a $d$-regular graph of girth~$g$~\cite{Sachs63},
	in particular there exists a $(d,g)$-cage
	(a graph with minimal order for the parameters~$d$ and~$g$).
	Every $(d,g)$-cage for an odd $d \geq 7$ is
	$\lceil \frac{d}{2} \rceil$-connected~\cite{BalbuenaS12}.
	
	Fix $k$
	and let $c_1$, $d$, $g$, and $i$ be the constants
	given by Lemma~\ref{LemmaLichterDetails} for $k+3$.
	Furthermore, let~$c_2$ be the constant given by 
	Lemma~\ref{LemmaCFICoprimeField} for $k+3$.
	We set $c := \max\set{c_1, c_2, k+3}$
	and assume that $d$ is odd (otherwise we increase $d$ by one).
	Let $G=(V,E,\leq)$ be a $(\max\set{d, 2c+1},g)$-cage
	(for an arbitrary order $\leq$).
	By Lemma~\ref{LemmaLichterDetails}
	there are functions $\lambda,\sigma\colon E\ra \Z_{2^i}$ such that $\sum\sigma=\sum \lambda + 2^{i-1}$ and 
	$\CFIgraphI{G}{\Z_{2^i}}{\lambda}\IMequiv{k+3}{\set{2}}\CFIgraphI{G}{\Z_{2^i}}{\sigma}$.
	By Lemma~\ref{LemmaCFICoprimeField}
	it holds that
	$\CFIgraphI{G}{\Z_{2^i}}{\lambda}\IMequiv{k+3}{\Primes \setminus \set{2}}\CFIgraphI{G}{\Z_{2^i}}{\sigma}$.
	The claim follows with Lemma~\ref{LemmaCombineCharacteristics}.
\end{proof}

\section{Conclusion}
There are two important conclusions that can be drawn from
Theorem~\ref{TheoremLimitationIMEquiv}.  The first is the immediate
one that there is no constant $k$ for which the $k$-invertible-map test
yields a complete isomorphism test.
\begin{corollary}
  There is no fixed $k$ such that $\equiv^{\textrm{IM}}_{k, \primes}$
  coincides with isomorphism on finite structures.
\end{corollary}

The CFI-structures constructed in the proof of Theorem~\ref{LichterTheorem}
are large: their size is super-exponential in $k$.
In particular, we get only a weak lower bound on $k$ in terms of the
size of the CFI-structures needed to distinguish them with the
invertible-map equivalence $\equiv^{\textrm{IM}}_{k, \primes}$.
The bound is super-constant but sub-logarithmic.
This should be contrasted
with the linear lower bound for the dimension of the Weisfeiler-Leman
method needed to distinguish the CFI-structures.  It is an interesting
question whether the bound for the invertible-map equivalence can be strengthened.

The second consequence is that no linear-algebraic logic captures
PTIME.  Indeed, the problem of determining, for a structure
$\CFIgraphI{G}{\Z_{2^i}}{\lambda}$, whether $\sum \lambda = 0$ is
decidable in polynomial time~\cite{Lichter21a}.
\begin{corollary}
No extension of fixed-point logic by linear-algebraic operators over fields captures PTIME.  
\end{corollary}

\bibliographystyle{plain}
\bibliography{DGLbib}

\end{document}